\documentclass[%
 aip,
 amsmath,amssymb,
reprint, onecolumn 
]{revtex4-1}

\usepackage{graphicx}
\usepackage{dcolumn}
\usepackage{bm}

\usepackage[utf8]{inputenc}
\usepackage[T1]{fontenc}
\usepackage{mathptmx}
\usepackage{etoolbox}

\makeatletter
\def\@email#1#2{%
 \endgroup
 \patchcmd{\titleblock@produce}
  {\frontmatter@RRAPformat}
  {\frontmatter@RRAPformat{\produce@RRAP{*#1\href{mailto:#2}{#2}}}\frontmatter@RRAPformat}
  {}{}
}%
\makeatother

\usepackage[utf8]{inputenc}
\usepackage[T1]{fontenc}
\usepackage{mathptmx}
\usepackage{etoolbox}
\usepackage{amsmath}
\usepackage{graphicx,color}
\usepackage{epstopdf}
\usepackage{epsfig}
\usepackage{mathrsfs}
\usepackage[breaklinks=true]{hyperref}

\newtheorem{theorem}{Theorem}

\newtheorem{definition}[theorem]{Definition}

\newtheorem{proposition}[theorem]{Proposition}

\newenvironment{proof}[1][Proof]{\noindent\textbf{#1.} }{\ \rule{0.5em}{0.5em}}

\begin{document}

\title{An Exact Formula for Quantum Entropy Production along Quantum Trajectories}

\author{John E.~Gough} \email{jug@aber.ac.uk}
   \affiliation{Aberystwyth University, SY23 3BZ, Wales, United Kingdom}

\author{Nina H. Amini} \email{nina.amini@l2s.centralesupelec.fr}
   \affiliation{Laboratoire des Signaux et Syst\'{e}mes (L2S), CNRS-CentraleSup\'{e}lec-Universit\'{e} Paris-Sud, Universit\'{e} Paris-Saclay, 3, Rue Joliot Curie, 91190, Gif-sur-Yvette, France.}

\begin{abstract}
We give an exact formula for the rate of change of the von Neumann entropy for the conditional state of a quantum system undergoing continuous measurement. Here we employ Paycha's Formula \cite{Paycha} which gives the noncommutative Taylor series development.
\end{abstract}

\maketitle

\section{Introduction}
The acquisition of information about a system from observations leads to a conditioning of prior probability distributions. In continuous time, the update of the state in this way is known as filtering and it is possible to study the corresponding entropy and make connections with thermodynamic concepts, see for instance Mitter and Newton \cite{Mitter_Newton}. here we examine the quantum analogue. The dynamical equation for the conditional density matrix of a quantum system undergoing continuous measurement takes the form \cite{Belavkin,BvHJ}
\begin{eqnarray}
    d \widehat{\rho_t} = \mathcal{L}^\star (   \widehat{\rho}_t ) \, dt + \sqrt{\eta} \, \big( L \widehat{\rho}_t  +  \widehat{\rho}_t L^\ast - \lambda_t  \widehat{\rho}_t \big) \, dI (t),
\label{eq:SME}
\end{eqnarray}
where $\mathcal{L} (\cdot ) $ is a generator of Lindblad type, $ \mathcal{L} (\cdot ) = - i[ \cdot , H ] + \sum_k \frac{1}{2}[L_k^\ast , \cdot ] L_k + \frac{1}{2} L_k^\ast [ \cdot , L_k]$, and $\mathcal{L} (\cdot ) ^\star$ is its adjoint (i.e., $\mathrm{tr} \{ \rho \, \mathcal{L} (X) \}\equiv \mathrm{tr} \{  \mathcal{L}^\star (\rho) \, X \}$ for all trace-class $\rho$ and bounded $X$), $H$ is the Hamiltonian of the system, $L$ is the coupling operator of the system to the monitored field (the \textit{collapse operator}). The term $I (t)$ is the \textit{innovations} (statistically a Wiener process) and $\lambda_t = \mathrm{tr} \{ \widehat{\rho}_t \, (L +L^\ast ) \}$.

The factor $\eta \in [0,1]$ is the \emph{efficiency} of the homodyne detector: $\eta = 1$ is perfect detection, while $\eta = 0$ effectively means that we are left with a dissipative dynamics described by a standard master equation.

The conditional state $\widehat{\rho}_t$ is a stochastic process adapted to the filtration generated by the measurements (or equivalently, the innovations process). Our interest shall be in the associated von Neumann entropies $\mathsf{S} \big(\widehat{\rho}_t \big) = - \mathrm{tr}_S \big( \widehat{\rho}_t \ln \widehat{\rho}_t \big) $ and this will likewise be an adapted stochastic process.

We recall that a state state is pure if and only if it is a rank-one projection; otherwise it is mixed. A mixed state is deemed to be \textit{faithful} if its density matrix $\rho$ is strictly positive. In such cases we may define its inverse $\frac{1}{\rho}$.

Our main result is an exact equation for the rate of change of the von Neumann entropy for the conditioned state.

\begin{theorem}
\label{thm:main}
On the assumption that the conditional state remains faithful, the von Neumann entropy for the conditioned state satisfies
 satisfies the equation
\begin{eqnarray}
    d \mathsf{S} \big(\widehat{\rho}_t \big)
    =
    - \mathrm{tr}_s \bigg(  \widehat{\rho}_t \,  \mathcal{L} \ln \widehat{\rho} _t \bigg) \, dt + \Sigma ( \widehat{\rho} _t) \, dt -    
     \mathrm{tr}_s \bigg(  \big( L \widehat{\rho} _t+  \widehat{\rho} _t L^\ast - \lambda (t) \widehat{\rho} _t  \big)
 \, \ln  \widehat{\rho}  _t\bigg) \, d I.
\label{eq:dS}
\end{eqnarray}
where
\begin{eqnarray}
    \Sigma (\widehat{\rho}) = \eta 
    \sum_{k_1, k_2=0}^\infty
     \dfrac{ 1  } {(k_1 + 1)  (k_1 + k_2 + 2)  }  \binom{k_1+k_2 }{ k_1} \,
   \, \mathrm{tr}_s \bigg\{
\big( - \frac{1}{\widehat{\rho} }\big)^{k_1+k_2 +1} 
\mathrm{ad}^{k_1}_{\widehat{\rho} } ( L \widehat{\rho}  +  \widehat{\rho}  L^\ast ) \,  \mathrm{ad}^{k_2}_{ \widehat{\rho } } ( L \widehat{\rho}  +  \widehat{\rho}  L^\ast)
    \bigg\} .
		\label{eq:Sigma}
\end{eqnarray}
\end{theorem}

We arrive at the form (\ref{eq:Sigma}) using a noncommutative Taylor series expansion formula due to S. Paycha \cite{Paycha}.

\subsection{Motivation}
Let us first remark that the stochastic master equation \ref{eq:SME} admits a pure state solution provided the following conditions are satisfied: the system is initially separated from its environment and in a pure state; the environment consists of just the detector (in which case the Lindblad generator is precisely $ \mathcal{L} (\cdot ) = - i[ \cdot , H ] + \frac{1}{2}[L^\ast , \cdot ] L + \frac{1}{2} L^\ast [ \cdot , L]$ where $L$ is the collapse operator); and we have perfect efficiency $(\eta = 1)$. Note that this situation generalizes to more than one monitored output field in the obvious way.

In this special situations, we have that $\widehat{\rho}_t \equiv | \psi_t \rangle \langle \psi_t |$ where the pure stochastic state satisfies (see the contribution of J. Kupsch to the standard reference book \cite{Decoherence})
\begin{eqnarray}
d | \psi_t \rangle = - \bigg( iH + \frac{1}{2} (L^\ast L - \lambda_t L + \frac{1}{4} \lambda_t^2 ) \bigg) | \psi_t \rangle \, dt
+ (L - \frac{1}{2} \lambda_t ) \, | \psi _t \rangle \, dI(t) .
\end{eqnarray}
The (conditional) von Neumann entropy, of course, vanishes in these cases.

An interesting special situation is where the collapse operator $L$ is self-adjoint. In particular, the question of whether the stochastic states converge to a pure eigenstates of the collapse operator become natural \cite{ABBH}. For instance, in \cite{SvHM} an ensemble of $N$ spin-1/2 particles was effectively studied through their total spin vector $\mathbf{J}$ and one performed a continuous homodyne monitoring with collapse operator $M\, J_z$ for some coupling constant $M$. Taking $H\equiv 0$ for convenience and assuming perfect detection efficiency, the conditional state will almost surely converge into one of the eigenstates of $J_z$, say $| m\rangle$, with probability $p_m = \langle m | \rho_0 | m \rangle $ in accordance with the Born rule. The conditional entropy should then converge to zero while the unconditioned should of course converge to $-\sum_m p_m \ln p_m$. Rather trivially, Theorem \ref{thm:main} would tell us that the rate of change of the conditional entropy would vanish for $\widehat{\rho}$ diagonal in the basis of the collapse operator. (In \cite{SvHM}, the problem is made more interesting by applying a feedback control through a Hamiltonian $\gamma \, b(t) J_y$ where $b(t)$ depends on the measurement record.)

A related problem arises in the theory of quantum feedback networks \cite{GouJam09a} \cite{GouJam09b}. One may model a quantum system with proportional feedback of homodyne measurement in a fairly straightforward way using the conditional dynamics \cite{Wiseman}. However, since one only needs to know that the measurement is being performed and the output fedback without actually knowing the measurement result, one can alternatively describe the average result coherently - see \cite{GouJam09b} section IV B. Here one is choosing to ignore the measurement record but the fact that a measurement is being made leads to back action - resulting in a unitary model for the system and the detector. This of course changes if one tried to look inside to see what the detector actually detects. An example where this may be of importance is the proposal for autonomous error-correction of a qubit using quantum switches \cite{KNPH}. Here several switches are coupled to the qubit in a coherent quantum feedback network: again, technically there is no measurement but one could potentially consider inserting detectors so the the qubit and the switches are driven by each others homodyne outputs. The set up in \cite{KNPH} is specifically designed so that the switches act as dissipators whenever a particular syndrome error occurs to the qubit. While the set up is completely coherent (therefore unitary), one may naturally ask if the energy dissipated by an error correction could leak out as heat thereby telling us that a error-correction had been made. In this example, the switches are acting as regulators giving a quantum analogue of the flyball governor. Related questions about quantum Maxwell's demon in the setting of superconducting qubits can be found in the work Cottet, et al. \cite{QMD}

\section{Discrete Time Models}

Our proof of Theorem \ref{thm:main} relies on formulating the problem in discrete time and taking the continuum limit. We begin with an exposition of the discrete time situation.

\subsection{State Evolution}
\subsubsection{Unconditional State Evolution}
We consider a discrete time open quantum dynamics with fixed time step $\tau >0$. Our system has Hilbert space $\mathfrak{h}_s$ and at each time $t=n \tau $ it comes in contact with a fresh ancillary system. The ancillas have a Hilbert space which is a copy of a fixed Hilbert space $\mathfrak{h}_a$ with a state which is a copy of a fixed state $\sigma$. The system state updates as
\begin{equation}
    \mathcal{V}_s \, \rho  = \mathrm{tr}_a \bigg( V ( \rho \otimes \sigma ) V^\ast \bigg) 
\end{equation}
where $V$ is a (copy of a) unitary operator on $\mathfrak{h} \otimes \mathfrak{h}_a$ and we have the partial trace over the ancilla. We note that $\mathcal{V}_s$ is a TPCP map.

If we start with system state $\rho_0$ the we obtain the sequence of states $(\rho_n)_n$ where $\rho_n= \mathcal{V}_s \rho_{n-1}$. That is $\rho_n = \mathcal{V}_s^n \rho_0$.


\subsubsection{Conditional State Evolution}
At each time $t=n \tau$ we can measure an observable $Y_n$ of the $n$th ancilla. We assume that $Y_n$ is a copy of a fixed observable $Y$ described by measurement operators $\{ M_y\}$ with $M_y^\ast M_y$ corresponding to the effect that $Y$ is measured to have value $y$: Note that $\sum_y M_y^\ast M_y =I_a$.

If the state of the system at a particular time is $\rho$ then we set
\begin{equation}
    \mathcal{M}_y \rho = \mathrm{tr}_a \bigg(  M_y V ( \rho \otimes \sigma ) V^\ast M_y^\ast \bigg).
\end{equation}
It follows that $\sum_y \mathcal{M}_y \equiv \mathcal{V}_s$. In particular, it is not a trace preserving map and we instead have
\begin{eqnarray}
p(y) \triangleq
    \mathrm{tr}_s ( \mathcal{M}_y \rho ) = \mathrm{tr}_{s+a} \bigg( M_y V (\rho \otimes \sigma ) V^\ast M_y^\ast \bigg)
\end{eqnarray}
and we interpret $p(Y=y)$ as the probability that we measure $Y$ to take the value $y$.
The conditioned state of the system immediately after $Y$ is measured to be $y$ is
\begin{eqnarray}
    \widehat{\rho} (y) = \dfrac{1}{p(y)} \mathcal{M}_y \, \rho .
\end{eqnarray}
One notes that $\sum_y  \widehat{\rho} (y) \,  p(y) = \mathcal{V}_s \rho$.


\subsubsection{Repeated Measurements}
The superoperators $\mathcal{V}_s$ and $\mathcal{M}_y$ describe the one-step processes. We now consider the situation where we initialize our system in state $\rho_0$ and at times $m\tau\, (m=1, \cdots n)$ measure an observable $Y_m$ which is a copy of the fixed $Y$ above, but with a fresh ancilla space. The measurement record will be a sequence of values $\mathbf{y} = (y_1, y_2, \cdots )$. We shall use the notation $ \mathbf{y}[m]$ for the subsequence $(y_1 , \cdots , y_m )$ for finite $m$.

For $n \ge m$, the state of the system at time $t=n\tau$ given the measurement record $ \mathbf{y}[m]$ is given by
\begin{eqnarray}
    \widehat \rho_{n|m} ( y_1, \cdots , y_m ) \triangleq  \dfrac{1}{p(y_1, \cdots , y_m ) } 
		\mathcal{V}_s^{n-m} \circ \mathcal{M}_{y_m} \circ \cdots \circ \mathcal{M}_{y_1} \rho_0  .
\end{eqnarray}
where the probability for the record is
\begin{eqnarray}
    p( y_1 , \cdots , y_m ) = \mathrm{tr}_s \bigg( \mathcal{M}_{y_m} \circ \cdots \circ \mathcal{M}_{y_1} \rho_0 \bigg).
\end{eqnarray}

For convenience, we will write $\widehat \rho_{n|m} ( y_1, \cdots , y_m )$ as $\widehat \rho_{n|m} ( \mathbf{y} )$ and $p( y_1 , \cdots , y_m ) $ as $p_m (\mathbf{y})$ even though the dependence is on $\mathbf{y}[m]$. We also set
\begin{eqnarray}
    \widehat{\rho}_n ( \mathbf{y} ) \triangleq \widehat{\rho}_{n|n} ( \mathbf{y} )
\end{eqnarray}
which is the estimate of the state at time $t=n \tau$ given all measurements up to that time. This is also called the \textit{filtered quantum state} at time $t=\tau n$.

We remark that we may also write
\begin{eqnarray}
    \widehat{\rho}_{n|m} ( \mathbf{y} ) =
      \dfrac{1}{p_m (\mathbf{y}) } \mathrm{tr}_{a^{\otimes n}} \bigg(
      V_n \cdots V_{m+1} M_{y_m} V_m \cdots M_{y_1} V_1  (\rho_0 \otimes \sigma^{\otimes n} )  \, V_1^\ast M_{y_1}^\ast \cdots V_m^\ast M_{y_m}^\ast V_{m+1}^\ast \cdots V_n^\ast \bigg) .
\end{eqnarray}
Here the $j$th measurement operator $M_{y_j}$ acts nontrivially on the $j$th ancilla space, while $V_k$ will couple the $k$th ancilla to the system and act trivially on all other factors. This implies that $[V_k, M_{y_j}]=0$ when $k>j$. This allows us to rewrite this as
\begin{eqnarray}
    \widehat{\rho}_{n|m} ( \mathbf{y} ) =
    \dfrac{1}{p_m (\mathbf{y})  } \mathrm{tr}_{a^{\otimes n}} \bigg(M_{y_m}  \cdots M_{y_1} U_n 
		(\rho_0 \otimes \sigma^{\otimes n} ) U_n^\ast M_{y_1}^\ast \cdots M_{y_m}^\ast \bigg) ,
\end{eqnarray}
where $U_n = V_n \cdots V_1$ (order is important!). We note that $[M_{y_j} , M_{y_k}] =0$ for $j \neq k$.


\subsection{Generic Model}
For definiteness, we take $\mathfrak{h}_a = \mathbb{C}^2$ with basis $|0\rangle =\left[\begin{array}{c}
     0  \\
      1
\end{array}
\right], \,
|1\rangle =\left[\begin{array}{c}
     1 \\
      0
\end{array}
\right]$. This is a qubit probe and we take the probe state before measurement to be $\sigma = | 0 \rangle \langle 0 |  \equiv \left[ \begin{array}{cc}
      0 & 0 \\
    0 & 1
\end{array}
\right] $.

It is convenient to represent operators on $\mathfrak{h}_s \otimes \mathfrak{h}_a$ as $2\times 2$ matrices whose entries are operators on $\mathfrak{h}_s$. In other words,
\begin{eqnarray}
    A\otimes | 1 \rangle \langle 1 | +
    B\otimes | 1 \rangle \langle 0 | +
    C\otimes  | 0 \rangle \langle 1 | +
    D\otimes  | 0 \rangle \langle 0 |
			\equiv
    \left[ \begin{array}{cc}
      A & B \\
    C & D 
\end{array}
\right]
.
\end{eqnarray}

We shall consider a one-step unitary of the form
\begin{gather}
    V = \exp \big\{ -i \tau \, H\otimes I_a \, + \sqrt{\tau}  L\otimes | 1 \rangle \langle 0 |- \sqrt{\tau} L^\ast |0 \rangle \langle 1 | \big\} \nonumber \\
\end{gather}
where $H=H^\ast$ and $L$ are bounded system operators. For our purposes, we will be interested in the small time step limit, $\tau << 1$, and to this end it suffices to work with the expansion
\begin{eqnarray}
    V \approx  \left[ \begin{array}{cc}
      I_s + \tau K & \sqrt{\tau} L \\
    -\sqrt{\tau} L^\ast  & I_s + \tau K
\end{array}
\right]
,
\end{eqnarray}
where $K= - \frac{1}{2} L^\ast L -iH $. We find that, for arbitrary system state $\rho$,
\begin{eqnarray}
    V (\rho \otimes \sigma ) V^\ast \equiv 
    V \left[ \begin{array}{cc}
      0 & 0 \\
    0 & \rho
\end{array}
\right] V^\ast 
\approx
\left[ \begin{array}{cc}
      \tau \, L \rho L^\ast  & \sqrt{\tau} \, L\rho \\
    \sqrt{\tau} \rho L^\ast & \rho + \tau \, (K \rho + \rho K^\ast  )
\end{array}
\right]
.
\end{eqnarray}
The unconditional system evolution is then
\begin{eqnarray}
\mathcal{V}_s \rho &=& \mathrm{tr}_a \bigg( V (\rho \otimes \sigma ) V^\ast \bigg) \approx \rho + \tau (L\rho L^\ast +K\rho +\rho K^\ast ) .    
\end{eqnarray}
Here we introduce the Lindblad generator 
\begin{eqnarray}
\mathcal{L} X = L^\ast XL- \frac{1}{2} L^\ast LX - \frac{1}{2}XL^\ast L -i [X,H]    \end{eqnarray}
and its adjoint $\mathcal{L}^\star \rho = L\rho L^\ast +K\rho +\rho K^\ast $ defined through the duality $\mathrm{tr}_s (\rho \, \mathcal{L}X ) = \mathrm{tr}_s ( \mathcal{L}^\star \rho \, X)$. With this, we may write
\begin{eqnarray}
    \mathcal{V}_s \rho \approx \rho + \tau \, \mathcal{L}^\star \rho .
\end{eqnarray}
The sequence $(\rho_n )$ generated by choosing $\rho_0$ and iterating according to $\rho_n = \mathcal{V}_s \rho_{n-1}$ can be seen as the approximation with time step $\tau$ to the continuous time $\rho (t)$ satisfying the master equation $\dfrac{d \rho}{dt }= \mathcal{L}^\star \rho$.

We now measure the observable $Z = \left[ \begin{array}{cc}
      0 & 1 \\
    1 & 0
\end{array}
\right] $
which has eigenstates $|\pm \rangle = \dfrac{1}{\sqrt{2}}  \left[ \begin{array}{c}
      1 \\
     \pm 1
\end{array}
\right]$. Equivalently, we have the pair of measurement operators
\begin{eqnarray}
    M_\pm =| \pm \rangle \langle \pm |
    \equiv \dfrac{1}{2}
    \left[ \begin{array}{cc}
      1 & \pm 1 \\
     \pm 1 & 1
\end{array}
\right]
\end{eqnarray}
and the associated maps are
\begin{eqnarray}
    \mathcal{M}_\pm \rho \approx \frac{1}{2} \rho  \pm \frac{1}{2} \sqrt{\tau} (L\rho + \rho L^\ast )
\end{eqnarray}
with probabilities $p(\pm 1) \approx \frac{1}{2} \pm \frac{1}{2} \sqrt{\tau} \mathrm{tr}_s (L\rho + \rho L^\ast )$.
In the Heisenberg picture, it is better to think of the measured observable as $Y=\sqrt{\tau}\,  V(I_s \otimes Z) V^\ast$ with $y= \pm \sqrt{\tau}$ as the two outcomes.

The measurements may be repeated leading to a measurement sequence $\mathbf{y}= (y_1, y_2, \cdots )$, say. The filtered quantum state $\widehat \rho_n ( \mathbf{y} ) =\dfrac{1}{p_n(\mathbf{y}) } \mathcal{M}_{y_n} \circ \cdots \circ \mathcal{M}_{y_1} \rho_0 $, satisfies the following difference equation (expanded to leading order in $\tau$)
\begin{eqnarray}
    \widehat{\rho}_n = \widehat{\rho}_{n-1} + \mathcal{L}^\star \widehat{\rho}_{n-1} \tau  + \bigg( L \widehat{\rho}_{n-1} + \widehat{\rho}_{n-1}L^\ast - \lambda_{n-1} \widehat{\rho}_{n-1}\bigg) \Delta I_n
    \label{eq:discrete_filter}
\end{eqnarray}
where $\lambda_{n-1} = \mathrm{tr}_s \big( \rho_{n-1} \, (L+L^\ast ) \big)$ and
\begin{eqnarray}
    \Delta I_n = y_n-y_{n-1} - \lambda_{n-1} \, \tau .
\end{eqnarray}
The terms $\Delta I_n$ are called the \textit{innovations} and have mean zero - in the continuous limit there will be a central limit effect which results in these being increments of a Wiener process: for our purposes, it is enough to note that to the appropriate order
\begin{eqnarray}
    ( \Delta I_n )^2 \approx \tau .
\end{eqnarray}
For a precise statement of this approximation, see Section 5.2 of \cite{BvHJ}.

The continuous limit for (\ref{eq:discrete_filter}) takes the form
\begin{eqnarray}
    d\widehat{\rho} (t) = \mathcal{L}^\star \widehat{\rho}(t) \, dt
    + \big( L \widehat{\rho}(t)+ \widehat{\rho}(t) L - \lambda (t) \widehat{\rho} (t) \big) \, dI(t) .
\end{eqnarray}
where now $\lambda (t) = \mathrm{tr}_s \big( \widehat{\rho} (t) \, (L+L^\ast )\big)$ and $I(t)$ is the limit Wiener process (the \textit{innovations process}). This is the well known form of the stochastic master equation for a homodyne measurement.


\section{Entropy Production}
Let $\rho$ be a quantum state (density matrix) which we write as $\rho= e^{-\Phi}$. With the standard notation $\langle A \rangle_\rho = \mathrm{tr} \big( \rho A \big)$, the von Neumann entropy of a system in state $\rho$ is given by 
\begin{eqnarray}
\mathsf{S}(\rho ) \triangleq \langle \Phi \rangle_\rho = - \mathrm{tr} (\rho \ln \rho )   .
\end{eqnarray}

The Kullback-Liebler divergence, or relative entropy, is $\mathsf{D} (\rho || \sigma ) \triangleq 
\mathrm{tr} \big( \rho \ln \rho - \rho \ln \sigma \big) $. The main properties that we 
shall use are: the \textit{Klein's inequality} which states $\mathsf{D} (\rho || \sigma )  \ge 0$ with equality if and only if $\rho = \sigma$; and the \textit{quantum data processing inequality} which states that $\mathsf{D} (\mathcal{N}\rho || \mathcal{N}\sigma ) \le \mathsf{D} (\rho || \sigma )$ for any trace-preserving CP map $\mathcal{N}$.

The unconditional entropy of the system at time $t=n \tau$ is
\begin{eqnarray}
    S_n = \mathsf{S} (\rho_n).
\end{eqnarray}
The conditional entropy at time $t= n \tau$ given the measurement record $\mathbf{y}[m]$ for $m \le n$ is $  \mathsf{S} \big( \widehat{\rho}_{n|m}  ( \mathbf{y} )\big) $, and we denote its average conditional entropy as
\begin{eqnarray}
    \overline{S}_{n|m} \triangleq \sum_{y_1 , \cdots , y_m} \mathsf{S} \big( \widehat{\rho}_{n|m}  ( \mathbf{y} )\big)  \, p_m (\mathbf{y}  ).
\end{eqnarray}

\begin{definition}
    The Holevo information at time $t=n \tau$ based on the first $m$ measurements is \cite{Holevo}
    \begin{eqnarray}
        H_{n|m} \triangleq S_n - \overline{S}_{n|m} .
        \label{eq:Holevo_info_def1}
    \end{eqnarray}
    The special case $H_n = H_{n|n}$ is just referred to as the Holevo information at time $t=n \tau$.
\end{definition}

\begin{proposition}
    The Holevo information may be written as
    \begin{eqnarray}
     H_{n|m} =\sum_{y_1 , \cdots , y_m}
    \mathsf{D} ( \widehat{\rho}_{n|m}  ( \mathbf{y} )\big) || \rho_n ) \, p_m( \mathbf{y} )  ,
    \label{eq:Holevo_info_def2}
\end{eqnarray}
and, in particular, satisfies $H_{n|m } \ge 0$.
\end{proposition}
\begin{proof}
We first observe that we may write
\begin{eqnarray*}
     \sum_{y_1 , \cdots , y_m}
    \mathsf{D} \big( \widehat{\rho}_{n|m}  ( \mathbf{y} )|| \rho_n \big) \, p_m( \mathbf{y} ) & =&
    \sum_{y_1 , \cdots , y_m} \bigg\{
    \mathrm{tr} \bigg( \widehat{\rho}_{n|m}  ( \mathbf{y} ) \, \ln  \widehat{\rho}_{n|m}  ( \mathbf{y} ) \bigg)
   	- \mathrm{tr} \bigg( \widehat{\rho}_{n|m}  ( \mathbf{y} ) \, \ln  \rho_{n}  \bigg)
    \bigg\} \, p_m( \mathbf{y} )  \nonumber\\
		&=& - \overline{S}_{n|m} -\sum_{y_1 , \cdots , y_m}  \mathrm{tr} \bigg( \mathcal{V}_s^{n-m}
    \circ \mathcal{M}_{y_m} \circ \cdots  \circ \mathcal{M}_{y_1} \rho_0 \, \ln \rho_{n}  \bigg) \\
    &=& - \overline{S}_{n|m} - \mathrm{tr} \bigg( \mathcal{V}_s^{n}
     \rho_0 \, \ln \rho_{n}  \bigg)\\
    &=& - \overline{S}_{n|m} - \mathrm{tr} \bigg(  
     \rho_n \, \ln \rho_{n}  \bigg) \equiv H_{n|m}.
\end{eqnarray*}
This establishes (\ref{eq:Holevo_info_def2}). Non-negativity then follows from Klein's inequality.
\end{proof}

\begin{definition}
    The incremental Holevo information is defined to be $\Delta H_n = H_n - H_{n-1}$.
\end{definition}

The next result is established by Landi, Paternostro, Belenchia in \cite{LPB_22}.
\begin{proposition}
The incremental Holevo information can be decomposed as
    \begin{eqnarray}
        \Delta H_n = G_n -L_n
        \label{eq:Holevo_info_increment}
    \end{eqnarray}
where 
    \begin{eqnarray}
        G_n  &\triangleq & H_{n|n} - H_{n|n-1},
        \label{eq:gain}\\
        L_n & \triangleq & H_{n-1|n-1} -H_{n|n-1}.
        \label{eq:loss}
    \end{eqnarray}
In particular, $G_n \equiv \overline{S}_{n|n-1}- \overline{S}_{n|n}$ while $L_n$ is strictly positive.

\end{proposition}
\begin{proof}
We sketch the argument given in \cite{LPB_22}.
The decomposition (\ref{eq:Holevo_info_increment}) follows from subtracting (\ref{eq:loss}) from (\ref{eq:gain}). Likewise, the form $G_n \equiv \overline{S}_{n|n-1}- \overline{S}_{n|n}$ follows immediately from the definition $H_{n|m} \triangleq S_n - \overline{S}_{n|m}$.

This leaves the positivity of $L_n$ to prove. We note that
\begin{eqnarray*}
    L_n &=&  \sum_{y_1 , \cdots , y_{n-1}}
    \bigg\{ 
    \mathsf{D} \big( \widehat{\rho}_{n-1|n-1}  ( \mathbf{y} ) || \rho_{n-1} \big) 
    - \mathsf{D} \big( \widehat{\rho}_{n|n-1}  ( \mathbf{y} ) || \rho_n  \big)
    \bigg\} \, p_{n-1}( \mathbf{y} ) ,
\end{eqnarray*}
however, using the quantum data processing inequality, we find
\begin{eqnarray*}
    \mathsf{D} \big( \widehat{\rho}_{n|n-1}  ( \mathbf{y} ) || \rho_n \big)
    &=&\mathsf{D} \big( \mathcal{V}_s \widehat{\rho}_{n-1|n-1}  ( \mathbf{y} )  || 
    \mathcal{V}_s \rho_{n-1} \big) \\
    &\leq & \mathsf{D} \big( \widehat{\rho}_{n-1|n-1}  ( \mathbf{y} )  || \rho_{n-1} \big)
\end{eqnarray*}
which leads to the desired result.
\end{proof}

\subsection{Entropy Inequalities}
Let $\rho_{ab}$ be a state with underlying Hilbert space $\mathfrak{h}_a \otimes \mathfrak{h}_b$ and let $\rho_a$ and $\rho_b$ be its marginals with respect to this tensor decomposition, i.e., $\rho_a = \mathrm{tr}_b ( \rho)$ and $\rho_b = \mathrm{tr}_a ( \rho )$. The mutual information is defined to be
\begin{eqnarray}
    \mathsf{I}_{a:b} (\rho) &\triangleq& \mathsf{D} ( \rho || \rho_a \otimes \rho_b ) \nonumber \\
    &\equiv & \mathsf{S}( \rho_a ) + \mathsf{S} (\rho_b ) - \mathsf{S} (\rho ) .
\end{eqnarray}
Clearly the mutual information is never negative.

We now return to our iterative interaction model.
Before the system and $n$th ancilla interact there joint state is $\rho_{n-1} \otimes \sigma$; afterwards, this is $V (\rho_{n-1}\otimes \sigma ) V^\ast$. The marginals of the latter, with respect to the system and ancilla tensor decomposition, are $\mathcal{V}_s (\rho_{n-1}) = \rho_n$ and 
\begin{eqnarray}
    \mathcal{F}_n (\sigma ) \triangleq \mathrm{tr}_s \big( V (\rho_{n-1}\otimes \sigma ) V^\ast \big) ,   
\end{eqnarray}
respectively. 

\begin{proposition}
    The entropy increment $\Delta S_n = S_n - S_{n-1}$ satisfies the inequalties
    \begin{eqnarray}
    \Delta S_n + \mathsf{S} ( \mathcal{F}_n \sigma ) -  \mathsf{S} (  \sigma )  \ge 0,
\end{eqnarray}
and
\begin{eqnarray}
    \Delta S_n +   \mathrm{tr}\big( (\sigma -\mathcal{F}_n \sigma )\, \ln \sigma \big) \ge 0.
\end{eqnarray}
\end{proposition}
\begin{proof}
The mutual information is therefore
\begin{eqnarray}
    \mathsf{I}_{s:a} \big( V (\rho_{n-1}\otimes \sigma ) V^\ast \big) 
		&=& \mathsf{S} (\rho_n ) + \mathsf{S} ( \mathcal{F}_n \sigma ) - \mathsf{S } \big( V (\rho_{n-1}\otimes \sigma ) V^\ast \big) \nonumber \\
    &=& \mathsf{S} (\rho_n ) + \mathsf{S} ( \mathcal{F}_n \sigma ) - \mathsf{S }( \rho_{n-1} ) - \mathsf{S} ( \sigma ).
\end{eqnarray}
Non-negativity of the mutual information then implies the first inequality.

Noting that $ 0 \le \mathsf{D} \big( \mathcal{F}_n \sigma || \sigma \big)  = - \mathsf{S} ( \mathcal{F}_n \sigma ) - \mathrm{tr}\big( \mathcal{F}_n \sigma \, \ln \sigma \big) $ and combining this with the first inequality, we arrive at the second inequality.
\end{proof}

\section{Explicit Formula}

We will make extensive use of the following formal noncommutative Taylor series due to S. Paycha \cite{Paycha}.

\begin{proposition}[Paycha's Formula]
    Let $f$ be analytic in a domain of the complex plane with $k$-th derivative $f^{(k)}$. Then for $\rho , \varepsilon$ noncommutative operators we have 
    \begin{eqnarray}
        f( \rho + \varepsilon ) = \sum_{n=0}^\infty f(n, \rho,\varepsilon ) ,
    \end{eqnarray}
    where the $n$th term in the Paycha expansion is
    \begin{eqnarray}
        f(n, \rho , \varepsilon) = \sum_{k_1, \cdots k_m=0}^\infty\dfrac{f^{(k_1+\cdots+k_n+n)} (\rho )} {(k_1+1)\cdots (k_1 + \cdots +k_n +n)}   \frac{{\mathrm{ad}^{k_1}_\rho (\varepsilon)}}{k_1!}\cdots 
    \frac{{\mathrm{ad}^{k_n}_\rho (\varepsilon)}}{k_n!}
    \end{eqnarray}
    where $\mathrm{ad}_\rho = [\rho , \cdot]$.
\end{proposition}

We note that $f(0, \rho , \varepsilon) \equiv f (\rho )$ while
\begin{eqnarray}
    f(1, \rho , \varepsilon ) = \sum_{k_1}^\infty
    \dfrac{f^{(k_1+1)} (\rho ) } {(k_1+1)!}
    \mathrm{ad}^{k_1}_\rho (\varepsilon)
\end{eqnarray}
and
\begin{eqnarray}
     f(2, \rho , \varepsilon) = \sum_{k_1, k_2=0}^\infty\dfrac{f^{(k_1+k_2+2)} (\rho )} {(k_1+1) (k_1 + k_2 + 2)}
    \frac{{\mathrm{ad}^{k_1}_\rho (\varepsilon)}}{k_1!}
    \frac{{\mathrm{ad}^{k_2}_\rho (\varepsilon)}}{k_2!}
    .
\end{eqnarray}
Our interest will be restricted to quantum diffusions appearing in homodyne measurement so we will not require the $n>2$ terms. However, this would not be the case for Poissonian processes encountered in photon number counting problems.


\subsection{The Unconditional Entropy}
As a first example of the noncommutative Taylor's expansion, we consider the unconditional entropy $S_n$. Let us set $f(z) =-z \ln z$, then the derivative terms are
\begin{gather}
    f^{(0)} (z) = - z \ln z, \, f^{(1)}(z) = - (1 + \ln z ), \nonumber \\
    f^{(m)}(z) = (m-2)!  ( -\frac{1}{z} )^{m-1} ,\, (m \ge 2).
\end{gather}

The state updates as $\rho_n = \rho_{n-1} + \mathcal{L}^\star \rho_{n-1} \tau$ and we use Paycha's formula with $\rho = \rho_{n-1}$ and $ \varepsilon = \mathcal{L}^\star \rho_{n-1} \tau$. The entropy update is
\begin{eqnarray}
    S_n &\approx& - \mathrm{tr}_s f ( \rho + \varepsilon ) \nonumber\\
    &\approx& - \mathrm{tr}_s f^{(0)} ( \rho ) +
    \sum_{k_1} \frac{1}{ (k_1 + 1)!} \mathrm{tr}_s \bigg( f^{(k_1+1)} (\rho)
    \, \mathrm{ad}^{k_1}_{\rho} (\varepsilon ) \bigg),
\end{eqnarray}
where we need only expand to the first order term. Moreover, we note that only the $k_1 =0$ term gives a non-zero contribution here leading to
\begin{eqnarray}
    S_n &\approx&  S_{n-1} - \mathrm{tr}_s \bigg( \rho_{n-1} \, \mathcal{L} (\ln \rho_{n-1} ) \bigg) \tau.
\end{eqnarray}
Here we have switched to the adjoint generator and used the fact that $\mathcal{L} ( I_s) = 0$.

The continuous limit form of this would be
\begin{eqnarray}
    \frac{d}{dt} S(t) = - \mathrm{tr}_s \bigg( \rho(t) \, \mathcal{L} \big( \ln \rho (t)\big) \bigg) ,
\end{eqnarray}
which agrees with the calculation in, for instance, \cite{Spohn}.

\subsection{Entropy along a Quantum Trajectory}
We now give a proof of Theorem \ref{thm:main}.

We will apply the same method to compute the update rule for $\mathsf{S} \big(\widehat{\rho_n} (\mathbf{y}) \big)$ for a fixed quantum trajectory generated by $\mathbf{y}$. We shall use the update rule (\ref{eq:discrete_filter}) with $\rho =\widehat{\rho}_{n-1}$ and $\varepsilon = \mathcal{L}^\star \widehat{\rho}_{n-1} \tau+ \sqrt{\eta } 
\bigg( L \widehat{\rho}_{n-1} + \widehat{\rho}_{n-1}L^\ast - \lambda_{n-1} \widehat{\rho}_{n-1}\bigg) \Delta I_n$.
This time, the bilinear terms in $\varepsilon$ are not negligible so we have to develop Paycha's formula to second order:
\begin{eqnarray}
    \mathsf{S} \big(\widehat{\rho_n} (\mathbf{y}) \big) 
    &\approx & - \mathrm{tr}_s f^{(0)} ( \rho ) 
		+ \sum_{k_1} \frac{1}{ (k_1 + 1)!} \mathrm{tr}_s \bigg\{ f^{(k_1+1)} (\rho)
    \, \mathrm{ad}^{k_1}_{\rho} (\varepsilon ) \bigg\} \nonumber \\
    &&
    +\sum_{k_1, k_2=0}^\infty
    \mathrm{tr}_s \bigg\{
		\dfrac{f^{(k_1+k_2+2)} (\rho )} {(k_1+1) (k_1 + k_2 + 2)}
    \frac{{\mathrm{ad}^{k_1}_\rho (\varepsilon)}}{k_1!}
    \frac{{\mathrm{ad}^{k_2}_\rho (\varepsilon)}}{k_2!}
    \bigg\} 
    .
\end{eqnarray}
The zeroth term in the Paycha expansion is $\mathsf{S} \big(\widehat{\rho}_{n-1}  (\mathbf{y}) \big) $. Here we note that $\mathrm{ad}^{k}_\rho (\varepsilon ) \approx \sqrt{\eta }  \mathrm{ad}^{k}_\rho ( L \rho + \rho L^\ast )\, \Delta I_n$. The first term in the Paycha expansion simplifies because only the $k_1 = 0$ term survives: all the higher terms involve a trace of a commutator so they vanish identically. This leads to
\begin{eqnarray}
    \mathsf{S} \big(\widehat{\rho_n} (\mathbf{y}) \big) 
    &\approx& \mathsf{S} \big(\widehat{\rho}_{n-1}  (\mathbf{y}) \big) 
    - \mathrm{tr}_s \bigg( \big(1+  \ln  \rho  \big) \varepsilon \bigg) \nonumber \\
    & +& \sum_{k_1, k_2=0}^\infty
    \dfrac{ \binom{k_1+k_2 }{ k_1}  } {(k_1 + 1)  (k_1 + k_2 + 2)  } 
    \mathrm{tr}_s \bigg\{
\big( - \frac{1}{\rho }\big)^{k_1+k_2 +1} 
\mathrm{ad}^{k_1}_\rho (\varepsilon ) \mathrm{ad}^{k_2}_\rho (\varepsilon)
    \bigg\}
    .
\end{eqnarray}

The result is that we obtain
\begin{eqnarray}
    \mathsf{S} \big(\widehat{\rho_n} (\mathbf{y}) \big) 
    &\approx&
    \mathsf{S} \big(\widehat{\rho}_{n-1}  (\mathbf{y}) \big) 	-
    \mathrm{tr}_s \bigg\{  \widehat{\rho}_{n-1} (\mathbf{y}) \mathcal{L} \ln \widehat{\rho}_{n-1} (\mathbf{y})    \big)  \bigg\} \tau \nonumber\\
    &&     -
 \mathrm{tr}_s \bigg\{  \big( L \widehat{\rho}_{n-1} (\mathbf{y})+\widehat{\rho}_{n-1}(\mathbf{y}) L^\ast - \lambda_{n-1}  \widehat{\rho}_{n-1} (\mathbf{y})  \big)
  \ln  \widehat{\rho}_{n-1} (\mathbf{y}) \bigg\} \Delta I_n
    \nonumber \\
   & +& \sum_{k_1, k_2=0}^\infty
     \dfrac{ \binom{k_1+k_2 }{ k_1}  } {(k_1 + 1)  (k_1 + k_2 + 2)  }     \mathrm{tr}_s \bigg\{
\big( - \frac{1}{\widehat{\rho}_{n-1} (\mathbf{y}) }\big)^{k_1+k_2 +1} \nonumber \\
 & &  \times \eta  \mathrm{ad}^{k_1}_{\widehat{\rho}_{n-1} (\mathbf{y}) } ( L \widehat{\rho}_{n-1} (\mathbf{y}) +  \widehat{\rho}_{n-1}(\mathbf{y}) L^\ast ) \mathrm{ad}^{k_2}_{\widehat{\rho}_{n-1} (\mathbf{y}) } ( L \widehat{\rho}_{n-1} (\mathbf{y}) +  \widehat{\rho}_{n-1}(\mathbf{y}) L^\ast)
    \bigg\} \tau
    .
\end{eqnarray}

At this stage, we see that the formula for the rate of change of entropy along a quantum trajectory in the continuous time limit will be (\ref{eq:dS}) with $\Sigma$ given by (\ref{eq:Sigma}).


\section{Conclusion}
In the obtained expression (\ref{eq:dS}) for the rate of change of the entropy along a quantum trajectory, the expression $\Sigma$ is clearly singular in the conditioned state $\widehat{\rho}$. This appears to be an intrinsic feature and faithfulness of the state is needed.

We note that $ \mathrm{ad}_\rho^k ( L \rho +\rho L^\ast ) = \mathrm{ad}_\rho^k ( L )\rho +\rho \, \mathrm{ad}_\rho^k (L^\ast )$ and that in general we may write
\begin{gather}
\mathrm{tr}_s \bigg\{
\big( \frac{1}{\widehat{\rho} }\big)^{\, k_1+k_2 +1} 
\mathrm{ad}^{k_1}_{\widehat{\rho} } ( L \widehat{\rho}  +  \widehat{\rho}  L^\ast ) \,  \mathrm{ad}^{k_2}_{ \widehat{\rho } } ( L \widehat{\rho}  +  \widehat{\rho}  L^\ast) 
    \bigg\}
		=
		\mathrm{tr}_s \bigg\{
 \widehat{\rho} ^{\, -k_1-k_2 -1} \, \mathrm{ad}^{k_1}_{\widehat{\rho} } ( L)  \widehat{\, \rho}^2   \mathrm{ad}^{k_2}_{\widehat{\rho} }   (L^\ast ) \bigg\} \nonumber\\
+ \mathrm{tr}_s \bigg\{  \widehat{\rho} ^{\, -k_1-k_2 } \bigg(  \mathrm{ad}^{k_1}_{\widehat{\rho} } ( L) \widehat{\rho}\, \mathrm{ad}^{k_2}_{\widehat{\rho} }   (L )+  \mathrm{ad}^{k_1}_{\widehat{\rho} } ( L^\ast)  \widehat{\rho} \, \mathrm{ad}^{k_2}_{\widehat{\rho} }  (L^\ast )
\bigg) \bigg\} \nonumber \\
+ \mathrm{tr}_s \bigg\{ \widehat{\rho}^{\, -k_1 -k_2 +1} \mathrm{ad}^{k_1}_{\widehat{\rho} } ( L^\ast) \mathrm{ad}^{k_2}_{\widehat{\rho} }  (L )
    \bigg\}
.
\end{gather}
By inspection, we see that $\Sigma$ is a sum of terms of the form $\mathrm{tr}_s \bigg\{ \widehat{\rho}^{-n+1} L_1 \widehat{\rho}^n L_2 \bigg\}$ for some $n=0,1,2, \cdots$, with various combinatorial coefficients and where either $L_1$ and $L_2$ are equal to $L$ or $L^\ast$.

The main result, Theorem \ref{thm:main}, gives the expression for the rate of change of entropy of the conditional state along a particular quantum trajectory under the assumption that the state is faithful. We give a series expansion for an entropy production rate $\Sigma$ which involves multiple nest commutators of the collapse operators with the conditional state. Our proof utilized a specific repeated qubit-probe model to approximate the measurement however the continuous limit should be universal.

\begin{acknowledgments}
This work is supported by the ANR project “Estimation et controle des syst\'{e}mes quantiques ouverts” QCOAST Projet ANR-19-CE48-0003, the ANR project QUACO ANR-17-CE40-0007, and the ANR project IGNITION ANR-21-CE47-0015.
\end{acknowledgments}


\end{document}